%% file: main.tex
\documentclass[runningheads]{llncs}
\usepackage[T1]{fontenc}
\usepackage{graphicx}
\usepackage{hyperref}
\usepackage{tikz}
\usetikzlibrary{calc,shapes,fit,arrows,positioning,decorations.pathreplacing}
\usepackage{amsxtra, amsfonts, amssymb, amstext}
\usepackage{mathtools}
\usepackage{nicefrac}
\usepackage{xspace}
\usepackage{cite}
\usepackage{thm-restate}

\usepackage{enumerate}
\usepackage{bm}
\usepackage{bbm}
\usepackage{xcolor}
\usepackage{todonotes}
\input{commands}

\begin{document}
%
% \MK{Let's also discuss the title at some point. We could e.g. stress that already "moderately social" vertices induce (strong) expanders in MCD-GIRGs}

\title{Expanders in Models of Social Networks}
\author{Marc Kaufmann \and Johannes Lengler \and Ulysse Schaller \and Konstantin Sturm}
\authorrunning{M. Kaufmann, J. Lengler, U. Schaller, K. Sturm}

\institute{Department of Computer Science, ETH Zürich, Zürich, Switzerland\\
\email{\{kamarc,lenglerj,ulysses,kosturm\}@ethz.ch}}
\maketitle              % typeset the header of the contribution
%
%\MK{Let's also discuss the title at some point. We could e.g. stress that already "moderately social" vertices induce (strong) expanders in MCD-GIRGs}
\begin{abstract}
    %\MK{Points to address in abstract. Work in progress, phrasing improvements actively invited at any time\begin{itemize}
        %\item Relevance of GIRGs and modeling real-world networks
       % \item short-falls of euclidean resp. generally norm-induced models
       % \item introduction and brief intuition behind minimum-component distance
        %\item introduction and (algorithmic) relevance of expanders
       % \item statement of our main result (and comparison to euclidean GIRGs)
  %  \end{itemize}}
    %In this paper, we show that any subgraph induced by vertices of weight at least $\log^{1+\varepsilon}n$, for any $\varepsilon>0$, in an MCD-GIRG of dimension $d\ge 2$, is an expander with high probability.
    A common model for social networks are Geometric Inhomogeneous Random Graphs (GIRGs), in which vertices draw a random position in some latent geometric space, and the probability of two vertices forming an edge depends on their geometric distance. The geometry may be modelled in two ways: either two points are defined as close if they are similar in \emph{all} dimensions, or they are defined as close if they are similar in \emph{some} dimensions.

    The first option is mathematically more natural since it can be described by metrics. However, the second option is arguably the better model for social networks if the different dimensions represent features like profession, kinship, or interests. In such cases, nodes already form bonds if they align in some, but not all dimensions.

    For the first option, it is known that the resulting networks are poor expanders. We study the second option in the form of Minimum-Component Distance GIRGs, and find that those behave the opposite way for dimension $d\ge 2$, and that they have strong expanding properties. More precisely, for a suitable constant $C>0$, the subgraph induced by vertices of (expected) degree at least $(\log n)^C$ forms an expander. Moreover, we study how the expansion factor of the resulting subgraph depends on the choice of $C$, and show that this expansion factor is $\omega(1)$ except for sets that already take up a constant fraction of the vertices. 
    This has far-reaching consequences, since many algorithms and mixing processes are fast on expander graphs.
    %MK{Regarding the last sentence, which of these implications (as well for spreading times of processes) we describe is still tbd.}
    %\MK{One issue for the implications is that many of the algorithmic consequences are only proven for expanders which are (as part of the expander definition) \textbf{bounded degree} (cf. e.g. kowalski's survey paper \url{https://people.math.ethz.ch/~kowalski/expander-graphs.pdf}), so not sure how much we can salvage here. See also the comment bubble on page 7}
    
\end{abstract}

% \MK{Things we should be consistent about
% \begin{itemize}
%     \item symbol for the power-law parameter $\tau$ vs. $\beta$
%     \item "w.r.t." vs. "wrt" and "w.h.p." vs. "whp" $\rightarrow$ always write "wrt" and "whp"
%     \item strip width vs height (and make sure to not repeatedly introduce the concept)
%     \item Minimum Component vs. Minimum-Component vs. their lower-case versions $\rightarrow$ always write "Minimum-Component Distance"
%     \item notation $uv$ (instead of $\{u,v\}$) for an edge
%     \item $\log^{\gamma}(n)$ vs. $\log^{\gamma}n$ $\rightarrow$ always use $\log^{\gamma}n$ 
%     \item $deg(.)$ vs. $\deg(.)$ $\rightarrow$ always use $\deg(.)$
%     \item Use proper quotation marks "..." vs ``...'' $\rightarrow$ always use ``...''
%     \item American vs. British English: e.g. rumor vs. rumour $\rightarrow$ use British English
%     \item Use dash for compound words e.g. well-known vs. En dash for clause separation: --
%     \item bibliography (JL: sorry, I inserted some papers without paying attention to the style that you have used so far.)
%     \item check if there is a \qed at the end of every proof
%     \item check that we don't write MCD distance (distance would be doubled)
% \end{itemize}}

%\footnotetext{Marc Kaufmann and Ulysse Schaller were supported by the Swiss National Science Foundation [grant number 200021\_192079]}

\section{Introduction}\label{sec:intro}

\input{introduction}

\section{Formal Definitions and Results}\label{sec:notation}
\input{notation_and_preliminaries}

%\section{Tools}\label{sec:tools}
%\input{chapters/tools}

\section{Proof Sketch}\label{sec:mcd_proof}
\input{proofsketch}

\begin{credits}
\subsubsection{\ackname} Marc Kaufmann and Ulysse Schaller were supported by the Swiss National Science Foundation [grant number 200021\_192079].
%
%\subsubsection{\discintname}
%The authors have no competing interests.
\end{credits}


\bibliographystyle{splncs04}
\bibliography{bibliography}

\appendix
\section{Appendix}\label{sec:appendix}
\subsection{Tools}\label{sec:tools_appendix}
\input{tools}
\subsection{Proof of the Main Theorem}\label{sec:proof_appendix}
\input{proof}

\end{document}

%% file: commands.tex
\DeclarePairedDelimiter\abs{\lvert}{\rvert}
\DeclarePairedDelimiter\norm{\lVert}{\rVert}

\newcommand{\br}[1]{\left(#1\right)}
\newcommand{\bre}[1]{\left[#1\right]}
\newcommand{\aS}{\abs{S}}

\newcommand{\ltg}{\log^{2\gamma}n}

\newcommand{\EE}[1]{\ensuremath{ \mathbb{E} \left[ #1 \right] } }
\newcommand{\Expected}[1]{\ensuremath{ \mathbb{E} \left[ #1 \right] } }
\renewcommand{\Pr}[1]{\ensuremath{\mathbb{P} \left(#1 \right) }}
\newcommand{\eps}{\varepsilon}

\newcommand{\Vol}{\ensuremath{\text{Vol}}}

\newcommand{\N}{\mathbb{N}}
\newcommand{\Z}{\mathbb{Z}}

\newcommand{\R}{\mathbb{R}}
\newcommand{\T}{\mathbb{T}}

\newcommand{\CC}{{\mathcal C}}

\newcommand{\CE}{{\mathcal E}}

\newcommand{\CG}{{\mathcal G}}
\newcommand{\CH}{{\mathcal H}}

\newcommand{\CV}{{\mathcal V}}

\newcommand{\pr}{\mathbb{P}}

%% file: introduction.tex
Personal attributes affect whom we know and whom we befriend. Most of our acquaintances are similar to us in some dimensions, for example, our colleagues share the same employer with us, our neighbours live in (almost) the same place, many of our friends share our hobbies, and so on. A common approach for generative models of social networks is thus to define a latent geometric space which captures similarity along such dimensions, and places edges (possibly in a randomized manner) between nodes that are geometrically close to each other.

A well-studied model with such a latent space are Geometric Inhomogeneous Random Graphs (GIRGs), which capture many structural properties observed in real networks. In particular, they are sparse (have a linear number of edges), small worlds, exhibit heavy-tailed degree distributions and have a large clustering coefficient~\cite{bringmann2024average}. Moreover, they are navigable~\cite{bringmann2022greedy}, match betweenness centrality of real networks~\cite{dayan2024expressivity}, reproduce the runtimes of various algorithms on real social networks~\cite{bläsius2024external}, and have been used to explain differences in epidemic curves in social networks~\cite{komjathy2024polynomial, komjáthy2024universalgrowthregimesdegreedependent} and the effectiveness of different epidemiological interventions~\cite{jorritsma2020interventions}. 

GIRGs are generated as follows. Each node $v$ independently draws a random position $x_v$ in a $d$-dimensional geometric ground space, where the different dimensions may be used to represent attributes like profession, place of residence, hobbies, kinship, and so on. Moreover, each node $v$ also independently draws a random weight $w_v$ from a power-law distribution, which will correspond to its expected degree, up to a constant factor. Every pair of nodes then forms an edge with a probability that increases with their weights and decreases with their geometric distance (the exact connection probability can be found in Definition~\ref{def:mcd_girgs}).

However, this construction leaves an important degree of freedom: how to measure the geometric distance in the ground space? Of course, an obvious possibility is to use a standard mathematical distance as induced by the Euclidean distance or the $L^\infty$ distance. However, those natural choices come with an unwanted side effect: in the Euclidean or $L^\infty$ distance, two points only have small distance if they are close \emph{along all dimensions}. If they differ strongly in a single coordinate, this suffices to make their geometric distance large. In terms of social networks, this is not desirable, since it would mean that two nodes who are similar in, say, hobby, residency and kinship, but not in their profession, would have a large geometric distance and thus a low probability to form an edge. This clearly does not match the distribution of bonds in social networks. 

For this reason, it was proposed in~\cite{bringmann2024average} to use the \emph{Minimum-Component Distance} (MCD) instead, where the geometric distance is given by the minimum distance along all dimensions.\footnote{MCD-GIRGs and Euclidean GIRGs have been jointly generalized by BDF-GIRGs, which inherit their realistic properties but offer a much richer spectrum of edge formation mechanisms by allowing distance functions with arbitrary combinations of minima and maxima~\cite{kaufmann2024sublinear}. For more information on BDF-GIRGs see Section~\ref{sec:related_work}.} Figure~\ref{fig:sample_mcd_girg} gives a realization of an MCD-GIRG in two dimensions. Note that the MCD is not a metric and does not obey the triangle inequality. For example, consider three points $x,y,z$ where $x$ and $y$ are close (have small MCD) because they agree along the first dimension, and $y$ and $z$ are close because they agree along the second dimension. Then this does not imply that $x$ and $z$ are also close. This matches the fact that nodes in social networks may belong to very different social circles. If $x$ and $y$ are co-workers, and $y$ and $z$ are close relatives, then it makes sense to consider $x$ and $y$ as ``close'' (for example, they probably have many shared neighbours in the network due to their shared work environment), and to also consider $y$ and $z$ as close, but this does not imply that $x$ and $z$ are close as well. 

Strikingly, it was shown in~\cite{bringmann2024average} that many of the fundamental properties of GIRGs do not depend on the underlying geometry. In particular, for \emph{any} geometric distance function\footnote{With some very mild technical conditions, in particular that the distance function is symmetric and that the volume of a ball of radius $r$ grows continuously with~$r$.}, the resulting model still has the \emph{ultra-small world} property that average graph distances (hop distances) are only of order $O(\log \log n)$, where $n$ is the number of vertices in the graph. Moreover, the model has a power-law degree distribution, and in both the MCD case and the case of classical Euclidean norms\footnote{Or the $L^\infty$ norm. In fact, the choice of the norm does not matter since all norms in $d$-dimensional space are equivalent.} it has high clustering coefficient $\Theta(1)$. None of these properties depend on the choice of the geometric distance function. However, in this paper we show a fundamental difference between the Euclidean GIRG model and the MCD-GIRG model: unlike the Euclidean version, MCD-GIRGs have strong \emph{expansion properties}.

\subsection{Our contribution}
% \medskip
% \noindent\textbf{Our contribution}
We show that in MCD-GIRGs, the subgraph induced by vertices of large degree forms an $\omega(1)$-expander. Intuitively, this means that, any vertex subset $S$ of this subgraph, with a mild restriction on the cardinality of $S$, has external neighbourhood whose cardinality is at least a $\omega(1)$-multiple of $|S|$. In the following, $\tau$ denotes the power-law exponent of the degree distribution, i.e., the number of vertices of degree larger than $x$ is $\Theta(nx^{1-\tau})$, where $n$ is the total number of vertices in the graph. We assume $2 < \tau < 3$, which is arguably the most common case for social networks, see Section \ref{sec:notation} for a discussion. For a graph $\CG = (\CV,\CE)$ and a vertex set $S\subseteq \CV$, we write $N_{ext}(S) \coloneqq \{ v \in \CV \setminus S : \exists u \in S \textnormal{ with } uv \in \CE\}$ for the external neighbourhood of $S$. Our main result can be summarized as follows.

\begin{theorem}[Main theorem, simplified]
    Let $\CG=(\CV,\CE)$ be an MCD-GIRG on $n$ vertices with dimension $d \ge 2$, power-law exponent $\tau \in (2,3)$ and $\alpha>1$. Let $\gamma > \frac{1}{3-\tau}$ and consider the subgraph $\CG'=(\CV',\CE')$ induced by vertices of %weight\footnote{The reader may for now just interpret the weights as the same as degrees: the weight is the expected degree up to a constant factor, and the actual degree of a vertex is tightly concentrated around its expectation. We discuss the details in Section \ref{sec:notation}.}
    degree at least $\log^{\gamma} n$. There exists a constant $\eps>0$ such that with high probability\footnote{\emph{With high probability} (whp) means that the probability converges to $1$ as $n \to\infty$.} the following holds. All subsets $S\subset \CV'$ of size $|S| \le \tfrac{1}{2}|\CV'|$ satisfy $|N_{ext}(S) \cap \CV'| \ge  \eps |S|$. Moreover, all subsets $\emptyset \neq S\subset \CV'$ of size $|S| = o(|\CV'|)$ satisfy $|N_{ext}(S) \cap \CV'| = \omega(|S|)$. 
    
    The same statement is true when $\CG'$ is replaced by the graph $\CG''$ induced by vertices of degree in the interval $[\log^\gamma n, 2 \log^\gamma n]$.%\JL{Would be nicer to have better names, at least $\CG_\gamma = (\CV_\gamma,\CE_\gamma)$ for $\CG'$, and also a different name for $\CG''$, but I am not sure that we will manage before the deadline. Can someone make a list of todos for after the deadline?}
    %\jl{Check whether we need a factor $\eps$ in the formula below, or whether we can swallow this in the $\gamma$.} \ks{Not 100 \% sure, but I believe we do need the $\eps$, as the size of $\abs{S}$ depends on $\gamma$ as well.}
    %\begin{align*}
    %    |N_{ext}(S) \cap \CV'| \ge  \eps \cdot |S|\cdot \min\left\{\log^{\gamma(3-\tau)}n, \br{\frac{|\CV'|}{\aS}}^{1-1/c_d}\right\}.
    %\end{align*}
\end{theorem}
In other words, the induced subgraph $\CG'$ is a vertex expander with an expansion factor $\omega(1)$ for all subsets $S$ whose size is $o(|\CV'|)$. The exact value of this expansion factor depends on the parameter $\gamma$ and thus on the subgraph that we consider. It is actually polylogarithmic if $S$ contains at most a polylogarithmic fraction of $\CV'$. See Theorem \ref{thm:main-thm} for the full statement containing the precise expansion factor we obtain. 
% \US{move the blue sentences below to the discussion after Theorem \ref{thm:main-thm} (?)}\blue{The larger weights/degree we consider, the better the expansion factor becomes. Eventually, when the size of $S$ gets close to $\CV'$, this expansion factor becomes impossible, simply because we run out of vertices. In this case, we prove a smaller expansion factor, but this factor is still larger than one for the whole range of $|S| \le \eps |\CV'|$.}  
%\JL{Something slightly ugly: this is still compatible with $\CG'$ having two giant component or something like that. It would be nicer to have a formulation which rules this out. With the deadline so soon, let's keep it this way for WG, but we should think about an alternative formulation.}
We mention here that the case of linear-sized subsets $S$ was already treated in~\cite{lengler2017existence}, as we discuss in more detail in Section~\ref{sec:mainresult}. The main contribution of this paper is to treat the case $|S|=o(|\CV'|)$ and in particular to show that in this case we have an expansion factor of $\omega(1)$.

\begin{figure}
    \centering
     \includegraphics[width=0.45\linewidth]{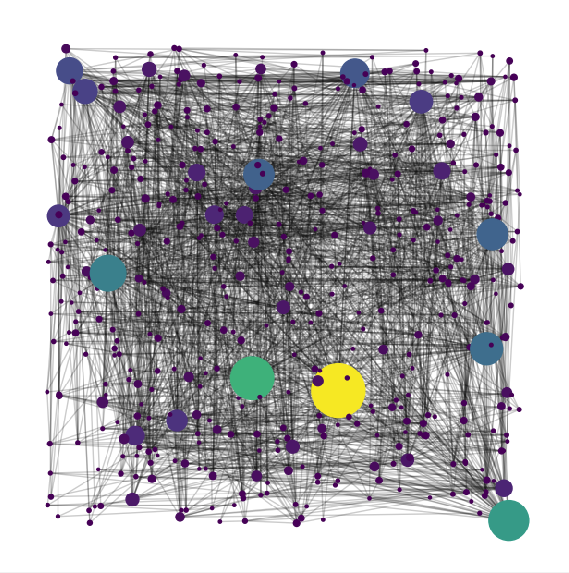}
    \includegraphics[width=0.45\linewidth]{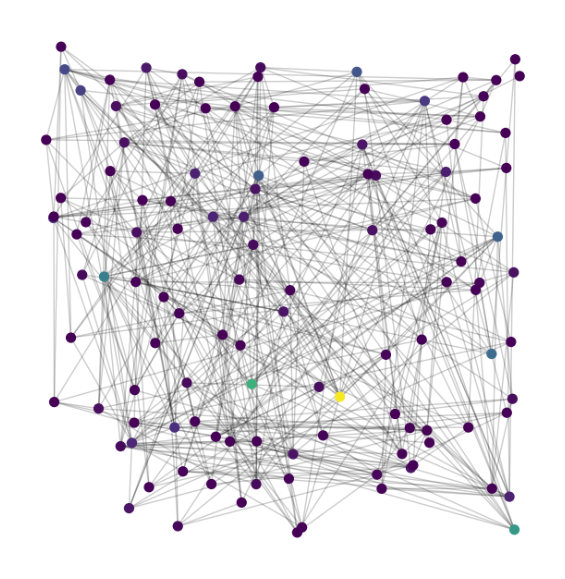}
   
    \caption{MCD GIRG (left), with parameters $d=2, n=500, \tau=2.1, \alpha=1.5$, and its induced subgraph of weights larger than $\log^2 n$ (right). Larger (and lighter-coloured) disks indicate nodes of larger weight/degree. Disks are only scaled on the left illustration to increase visibility.}
    \label{fig:sample_mcd_girg}
\end{figure}

\subsection{Implications}
%\JL{Write here the implications of the expansion property.}
%\JL{Mention here (or even earlier) that actually the same statement version holds for the subgraph induced by vertices of weight in $[w_0,2w_0]$.}
Expander graphs have an astonishing wealth of applications across mathematics and computer science -- and consequently, MCD-GIRGs inherit a rich list of such useful and desirable properties. For an introduction to the field, we refer the reader to the survey papers \cite{Hoory2006expander,lubotzky2011expandergraphspureapplied}. We discuss some of the most important applications and properties of expanders in the following paragraphs.
\paragraph{Remark.} Some of the applications described below assume that the expander graph in question is \emph{almost regular}, by which we mean that the minimum and maximum degree in the graph only differ by a constant factor. But since our result also holds in the graph $\CG'' = (\CV'',\CE'')$, and in $\CG''$ the minimum and maximum degree (induced in $\CG''$) are easily seen to differ by at most factor of $2+o(1)$, this subgraph is indeed regular up to a constant factor. We also mention that $\CV''$ is large, $|\CV''| = n/\log^{O(1)}n$, and the whole giant component of $\CG$ is well-connected to $\CV''$. Whp, all vertices of degree larger than $2\log^\gamma n$ have neighbours in $\CV''$, and it was shown in the so-called \emph{bulk lemma} in~\cite{bringmann2024average} that whp every vertex in $\CV$ in the giant component has a path of at most polylogarithmic length to a vertex of degree at least $\log^\gamma n$.
%\MK{In progress.}
\begin{itemize}
    \item \textit{Rapid Mixing of Random Walks:} On almost regular expanders, random walks will quickly reach a near-stationary distribution. This allows for random sampling of vertices  (within the induced subgraph), even if the graph is only accessible locally, which can be used for representative polling and computation of expectations of vertex properties. A practical application of random walks is the PageRank algorithm, which is used by search engines on the web graph. If the graph has expansion properties then PageRank converges quickly. Rapid mixing may also be beneficial for property testing: sampling a small number of vertices and their neighbourhoods may demonstrate more efficiently whether a given graph has some local property -- or at least has a large subgraph with this property.
    \item \textit{Distributed Computing:} The rapid mixing property is very useful in distributed computing. Algorithms that rely on message passing between nodes will converge quickly on this type of graph. One example is best-of-$k$ voting (for any $k\ge 2$), where each node randomly samples $k$ neighbours and updates its opinion to the majority opinion among them: On certain expander graphs, with best-of-$k$ voting, consensus is reached in $O(\log n)$ rounds \cite{shimizu2024quasi}. This matches the asymptotic consensus time of best-of-three voting on complete graphs~\cite{ghaffari2018nearly}. 
    % the result is for $\lambda$-expanders with $\lambda O(n^{-1/4})$, where lambda is an upper bound on both the absolute values of both the second-largest and the smallest eigenvalue of the transition matrix of the simple random walk on the graph.
    In graphical models, belief propagation, a message-passing algorithm for inference, often converges quickly on graphs with good expansion. Similarly, algorithms like distributed averaging and load balancing may benefit from good expansion. This has great potential for the design of distributed algorithms of this type for MCD-GIRGs, where the algorithm is first run on the induced subgraph $\CG'$ (or the almost regular $\CG''$), where all these theoretical performance guarantees are present.
    \item \textit{Robustness to Edge or Vertex Failures:} Expanders are robust to random failures. Even after removing a random small constant fraction of edges or vertices, the remaining graph is still well-connected. Since expanders offer several short paths, low-congestion and robust routing are possible~\cite{racke2002minimizing}.
    %The high-degree expander core provides resilience.
    \item \textit{Information Dissemination:} Information -- or a rumour or disease -- can spread very quickly through the graph. We get for free that a rumour will spread to most large-degree vertices in polylogarithmically many rounds, since the expansion property ensures rapid spreading through the subgraph of vertices with degree $[\log^\gamma n,2\log^\gamma n]$, as well as to all their neighbours. For the same reason, in an SI model the infection will spread among the large-degree vertices in $o(\log n)$ rounds.

    \item \textit{Fast Algorithms for Graph Problems:} Problems such as determining connectivity or finding shortest (approximate) paths %, and computing minimum cuts 
    can often be solved more efficiently on expanders or graphs with expander-like components. 
   %\item \textit{Derandomization:} Expanders are a fundamental tool in the derandomization of randomized algorithms, which allows one to obtain deterministic algorithms with similar performance guarantees. 
   % \item \textit{Expander Decompositions:} Expander decompositions - as well as the related low-diameter decompositions - of graphs have found many applications in graph algorithms, enabling powerful and flexible paradigms for close-to-linear-time algorithms ~\cite{spielman2004nearly,haeupler2024new,bringmann2025nearoptimaldirectedlowdiameterdecompositions,saranurak2019expander, kelner2013almost}.
   %\JL{I have removed the last two because that is not a something on GIRGs, but it is something where one finds an arbitrary expander and uses that. There we don't really get an advantage from knowing that GIRGs are expanders.}
    More generally, many \textit{NP-hard problems}, such as graph colouring, finding large independent sets or balanced separators, minimum bisection, maximum cut, as well as determining the graph bandwidth, have good approximation algorithms on expanders~\cite{Chuzhoy2019ADA}. Some of these approximation algorithms are among the most central and widely used tools in algorithm design, especially due to their natural connections to the hierarchical divide-and-conquer paradigm~\cite{racke2002minimizing}.
    % \begin{itemize}
    % \item Graph Coloring:
    % \item Finding Independent Sets:
    % \item Graph Bandwidth:
    % \item Max Cut:
    % \item Balanced Separators:
    % \item Minimum Bisection:
    % \item Sparsest Cut:
    % \end{itemize}
    \item Expander graphs also appear in many other algorithms. In streaming algorithms, when processing a stream of edges representing the graph, expander properties allow to maintain certain graph characteristics such as connectivity using limited memory. 
\end{itemize}
\subsection{Related work}
\label{sec:related_work}
The GIRG model with classical Euclidean geometry has been intensively studied over the last decade. A popular predecessor model are \emph{Hyperbolic Random Graphs} (HRG)~\cite{krioukov2010hyperbolic}, which turned out to be a special case of GIRGs~\cite{bringmann2019geometric}. Although HRG can be defined for higher dimensions, they are traditionally only studied in the one-dimensional case $d=1$.\footnote{In the HRG terminology this still involves a two-dimensional disk since the weight is encoded as one dimension. Nevertheless, it corresponds to the $d=1$ case.} Note that in this case there are no differences between Euclidean geometry and MCD. The GIRG model was invented independently, with only minor technical deviations, as \emph{Scale-Free Percolation} (SFP)~\cite{deijfen2013scale} and has been studied under this name as well. See~\cite{komjathy2020explosion} for a comparison of SFP and GIRG. 

A particularly important result on Euclidean GIRGs is that they have small separators: if we split the geometric space into two halves via a hyperplane, then there are only $O(n^{1-\eps})$ edges crossing this hyperplane, for some $\eps >0$ that depends on the model parameters~\cite{bringmann2019geometric}. This implies that it is possible to split the giant component (of either $\CG$ or $\CG'$) into two components of equal size by removing $O(n^{1-\eps})$ edges from the graph. Those separators exist at all scales. Consider any cube $\CC$ containing $x$ vertices, then the expected number of edges crossing the boundary of this cube is $O(x^{1-\eps})$. In particular, this yields a vertex set of size $x$ whose exterior neighbourhood has size at most $O(x^{1-\eps})$, which is asymptotically smaller than $x$ if $x\to\infty$. Not only is this construction possible for any growing function $x = \omega(1)$, but we may also place the cube $\CC$ at \emph{any} position in the geometric space. This means that small separators and non-expanding sets are ubiquitous in Euclidean GIRGs.

While there is ample work on the models with Euclidean geometry, work on substantially different geometries is still sparse. The first paper to consider this was~\cite{bringmann2024average}, where the MCD-GIRG model was introduced as an example for a non-metric distance function and it was shown that the power-law degree sequence, the $O(\log \log n)$ average hop distance, and the poly-logarithmic diameter are independent of the underlying geometry. The first substantial structural difference between MCD-GIRGs and Euclidean GIRGs was shown in~\cite{lengler2017existence}, where it was proven that the giant component of MCD-GIRGs does not have sublinear separators, in contrast to the Euclidean case. In other words, every vertex subset of linear size of the giant component has an exterior neighbourhood of size $\Omega(n)$ in MCD-GIRGs. This result is both weaker and stronger compared to our result: it is stronger because it includes all vertices, also vertices of smaller weight, as long as they are in the giant component. But it is weaker because it only covers very large sets of size $\Theta(n)$, not smaller sets. This makes it less useful for applications, since expanders need that vertex sets at all scales are expanding. Moreover, the expansion factor was naturally only constant in~\cite{lengler2017existence}, whereas we obtain a much larger expansion factor unless the vertex set has linear size.

The MCD was extended in~\cite{kaufmann2024sublinear} to general combinations of maxima and minima. The resulting distance functions were called \emph{Boolean Distance Functions} (BDF), and the corresponding network model BDF-GIRGs, which inherits the desirable properties of real-world networks captured by Euclidean GIRGs such as sparsity, small diameter and ultra-small average distances, and jointly generalizes MCD-GIRGs and Euclidean GIRGs. An example of such a BDF would be $dist \coloneqq \min\{dist_{work}, \max\{dist_{hobby}, dist_{residence} \}\}$, indicating that two individuals are likely to know each other if they are either working in closely related fields, or if they share both the same hobby and live in the same place. There are two main results in~\cite{kaufmann2024sublinear}. %The first one is a complete classifications of BDFs into Euclidean-like and MCD-like formulas. 
The first one is a complete classification of the occurrence of sublinear separators. A sub-category, \emph{single-coordinate outer max} (SCOM) GIRGs behave as Euclidean GIRGs in this regard and have sublinear separators of the giant component. All other BDF-GIRGs behave like MCD-GIRGs and have no such sublinear separators. This result from~\cite{kaufmann2024sublinear} thus extended the analysis from~\cite{lengler2017existence} to arbitrary BDFs. The second result in~\cite{kaufmann2024sublinear} is that all BDFs satisfy the following \emph{stochastic triangle inequality}. Consider a radius $r$ and a point $x$ in a geometric space with some BDF distance function $d$. If we sample two points $y,z$ uniformly at random from the ball of radius $r$ around $x$, then $\Pr{d(y,z) \le 2r} = \Omega(1)$, independently of $r$. Note that the probability is $1$ when $d$ is a metric, but recall that BDFs are in general not metrics, e.g.\ the MCD is not a metric. It was shown in~\cite{kaufmann2024sublinear} that this stochastic triangle inequality suffices to guarantee that all BDF-GIRGs exhibit strong clustering, more precisely a clustering coefficient of $\Theta(1)$. 

It is natural to ask whether the same characterization also applies to separators of all scales, i.e.\ whether the result of this paper generalizes to all BDFs which are not SCOM. However, our proof methods do not easily generalize to the whole class of BDFs, so this remains an open problem.

Another aspect in which Euclidean and MCD-GIRGs differ is rumour spreading. While our result immediately implies that rumours need at most polylogarithmic time to spread through an MCD-GIRG (because it takes polylogarithmic time to reach the first vertex in $\CV''$, then polylogarithmic time to reach \emph{every} vertex in $\CV''$, and finally polylogarithmic time to reach all other vertices from $\CV$), more precise results were obtained in~\cite{kaufmann2024rumour}. There it was shown that there are three different speeds at which a rumour can reach a constant fraction of the vertex set: \emph{ultra-fast} in time $O(\log \log n)$; \emph{fast} in polylogarithmic time; or \emph{slow} in time $n^{\Omega(1)}$. The latter possibility had first been observed in a slightly different model, \emph{Spatial Preferential Attachment} (SPA)~\cite{feldman2017high}. Note that our result rules out slow rumour spreading in MCD-GIRGs. However, it has already been shown in~\cite{kaufmann2024rumour} that all three regimes can occur in Euclidean GIRGs depending on the model parameters, but that rumour spreading is always ultra-fast in MCD-GIRGs (assuming the power-law exponent $\tau$ is in the interval $(2,3)$). This coincides with the non-geometric setting of Chung-Lu graphs, where rumours also spread ultra-fast when $2<\tau<3$~\cite{fountoulakis2012ultra}.

Finally, an efficient algorithm for sampling MCD-GIRGs in time $O(n)$ was recently presented in~\cite{dayan2024expressivity}. Even though several algorithms had been developed for sampling Euclidean GIRGs in linear time~\cite{bringmann2019geometric,blasius2022efficiently}, those algorithms are fundamentally based on the existence of small separators in the underlying geometric space. Hence, those algorithms could not be transferred to MCD-GIRGs. The algorithm~\cite{dayan2024expressivity} closes this gap and allows to efficiently sample and study MCD-GIRG empirically.

\paragraph{Organization of the paper}
The remainder of the paper is organized as follows. In Section~\ref{sec:notation} we formally introduce MCD-GIRGs and give the formal statement of our main result. %Moreover, we provide some important notation and terminology. 
Section~\ref{sec:mcd_proof} then gives a proof overview and intuition behind the overall strategy The formal proof of our main result can be found in Appendix~\ref{sec:proof_appendix}, preceded by some useful technical tools in Appendix~\ref{sec:tools_appendix}.

%% file: notation_and_preliminaries.tex
\paragraph{Notation and Terminology.}
We write $[n] \coloneqq \{1,..., n\}$. The $i$th component of a vector $e$ is denoted by $e_i$. In other words, we will denote the $i$th component of the position of a vertex $v$ by $x_{v,i}$ for $i \in [d]$. We denote the degree of a vertex $v$ constrained to a set of other vertices $S$ by $\deg_S(v)$. We denote by $uv$ the edge connecting two vertices $u$ and $v$. We say that an event $\mathcal{A}$ occurs with high probability if $\pr[\mathcal{A}] = 1 - o(1)$ as the number of vertices $n$ tends to $\infty$. We denote the natural logarithm of $x$ by $\log{x}$.

\subsection{(MCD-)GIRGs.}
As our ground space, we consider the $d$-dimensional torus $\T^d \coloneqq \R^d / \Z^d$. This can be viewed as the $d$-dimensional unit cube $\bre{0,1}^d$, with opposite faces identified, which has the advantage of yielding a bounded and symmetric ground space. In this topology, the absolute difference for $a, b \in [0,1]$ is defined as $\abs{a-b}_T \coloneqq \min\{\abs{a-b}, 1- \abs{a-b}\}$. As we will only deal with the torus topology, we will simply write $\abs{x}$ instead of $\abs{x}_T$.

We now define the distance function, which we will use to construct the graphs. Consider $x_u, x_v \in \T^d$ with $x_u = (x_{u,1}, ..., x_{u,d})$ and $x_v = (x_{v,1}, ..., x_{v,d})$. We define the Minimum-Component Distance (MCD) as
\begin{align*}
    \norm{x_u-x_v}_{\min} \coloneqq \min\{\abs{x_{u,i}-x_{v,i}}_T \mid 1\le i \le d \}.
\end{align*}
This distance function gives rise to the measurable volume function $V_{\min}: \R_{0}^{+} \rightarrow [0,1]$ given by $V_{\min}(r) \coloneqq \Vol\br{\{x \in \T^d \mid \abs{x} \leq r\}}$, which is symmetric, translation-invariant and continuous. For $r \rightarrow 0$, the volume function satisfies $V_{\min}(r) = \Theta(r)$ \cite{lengler2017existence}. See also  Figure~\ref{fig:illustrations} (left) for an illustration of ``balls'', which are cross-shaped in this geometry. 
The distance and the volume function together allow us to define Geometric Inhomogeneous Random Graphs (GIRGs).

GIRGs have originally been introduced in \cite{bringmann2019geometric}, though only for Euclidean distances, and have been generalized to arbitrary geometries in~\cite{bringmann2024average}. We begin by formalizing the concept of a power-law distribution.  

\begin{definition}\label{def:power-law}
    % Let $\CV$ be a set of vertices and $(w_v)_{v \in \CV}$ a sequence of weights associated with this set. Let $\Bar{w} = \Bar{w}(n)$ satisfy $n^{\omega(1/ \log{}\log{n})} \leq \Bar{w} \leq n^{(1- \Omega(1))/(\tau -1)}$. The sequence  $(w_v)_{v \in \CV}$ follows a power-law distribution with exponent $\tau \in (2, 3)$ if the following two conditions are satisfied.
    %
    % \begin{enumerate}
    %     \item The minimum weight is constant, \ie, 
    %     \begin{equation}\tag{SP1}
    %         w_{min} \coloneqq \min\{w_v \mid v \in \CV\} = \Omega(1)
    %     \end{equation}
    %     \item For all constants $\eta > 0$ there are $c_1, c_2 > 0$ such that
    %     \begin{equation}\tag{SP2}
    %         c_1 \frac{n}{w^{\tau -1 + \eta}} \leq \#\{v \in \CV \mid w_v \ge w\} \leq c_2 \frac{n}{w^{\tau -1 - \eta}}
    %     \end{equation}
    %     holds for sufficiently large n. The first inequality only holds for $w_{min} \leq w \leq \Bar{w}$ and the second one holds for all $w \ge w_{min}$
    % \end{enumerate}
    Let $\tau>1$. A discrete random variable $X$ is said to follow a \emph{power-law with exponent} $\tau$ if $\pr(X = x) = \Theta(x^{-\tau})$ for $x\in \N$. A continuous random variable $X$ is said to follow a \emph{power-law with exponent} $\tau$ if it has a density function $f_X$ satisfying $f_X(x) = \Theta(x^{-\tau})$ for $x\ge 1$.
\end{definition}

Now, we are ready to formally define MCD-GIRGs. Throughout the paper, we will assume that the three parameters $\tau$, $\alpha$ and $d$ below are constants.

\begin{definition}[MCD-GIRGs]
\label{def:mcd_girgs}
    % Let $d \in \N$, $\alpha > 1$ and $(w_v)_{v \in \CV}$ a sequence of weights following a power law with exponent $\tau \in (2,3)$ as in Definition \ref{def:power-law}. A MCD-Geometric Inhomogenous Random Graph (MCD-GIRG) $\CG= (\CV, \CE)$ is a simple undirected graph on $n$ vertices and is obtained in the following fashion.
    % \begin{enumerate}
    %     %\item Every vertex $v \in \CV$ draws independently a weight $w_v$ from the distribution $\CD$.
    %     \item Every vertex $v \in \CV$ draws independently uniformly at random a position $x_v$ in the torus $\T^{d}$.
    %     \item Each pair of distinct vertices $u,v \in \CV$ is independently connected by an edge with probability $p_{uv}$ with 
    %     \begin{align}
    %         p_{uv} &\ge c_L \cdot \min\left\{\frac{w_u w_v}{n \cdot V_{min} (\norm{x_u -x_v}_{min})}, 1\right\}^{\alpha} \text{and} \tag{EPL}\label{eq:EPL}\\
    %         p_{uv} &\leq c_U \cdot \min\left\{\frac{w_u w_v}{n \cdot V_{min} (\norm{x_u -x_v}_{min})}, 1\right\}^{\alpha}\tag{EPU}\label{eq:EPU}
    %     \end{align}
    % \end{enumerate}
 Let $\tau>2$, $\alpha>1$ and $d\in\mathbb{N}$ and let $\mathcal{D}$ be a power-law distribution on $[1,\infty)$ with exponent $\tau$. A \emph{MCD-Geometric Inhomogeneous Random Graph (MCD-GIRG)} with vertex set $\mathcal{V}=[n]$ and edge set $\mathcal{E}$ is obtained by the following three-step procedure:

 \begin{enumerate}
    \item Every vertex $v\in\mathcal{V}$ draws i.i.d.\ a \emph{weight} $w_v \sim \mathcal{D}$.
     
     \item Every vertex $v\in\mathcal{V}$ draws independently and u.a.r.\ a position $x_v\in \T^{d}$.

      \item For every two distinct vertices $u,v \in\mathcal{V}$, add the edge $uv\in\mathcal{E}$ independently with probability
      \begin{align}
      \label{eq:connection}
         \mathbb{P}[uv\in\mathcal{E} \mid w_u, w_v, x_u, x_v] = \Theta\Big(\min\Big\{\frac{w_uw_v}{n \cdot V_{min}(\|x_u-x_v\|_{min})}, 1\Big\}^{\alpha}\Big),
      \end{align}
      where the hidden constants are uniform over all $u,v$.
 \end{enumerate}  
\end{definition}

With the connection probability given above, the expected degree of a vertex $v$ is of the same order as its weight $w_v$ (Lemma 4.3 in~\cite{bringmann2024average}) and its actual degree is a sum of independent random variables which converges to the Poisson distribution in the limit. The restrictions on the range of the power-law parameter $\tau$ come from the observation that the degrees in many real networks, in particular social networks, approximately follow a power-law with $\tau$ in this range~\cite{voitalov2019scale}. In particular, such a power-law distribution has infinite second moment. Nonetheless, the resulting GIRGs are sparse, that is, have only a linear number of edges. We call an edge a \emph{strong tie} if the minimum in~\eqref{eq:connection} is taken by $1$, and \emph{weak tie} otherwise. The latter are edges that exist although they are unlikely given the weights and positions of the vertices. The parameter $\alpha$ quantifies the influence of the geometry in the sense that it controls the number of weak ties (increasing $\alpha$ reduces the number of weak ties). These weak ties are an important concept in sociology~\cite{granovetter1973weak} and influence the spreading of rumours~\cite{kaufmann2024rumour} and viral infections~\cite{komjathy2024polynomial,komjáthy2024universalgrowthregimesdegreedependent,jorritsma2020interventions}. Strong ties, on the other hand, are edges which have at least a constant probability of existing. The geometric region around a vertex $v$ of the radius $r$ that satisfies $V_{\min}(r) = w_v/n$, is also referred to as the \emph{ball of influence} of $v$. It is the region around $v$ where all edges are strong ties, regardless of the weight of the other endpoint.
%We also note here that, while the Minimum-Component Distance does not fulfil the triangle inequality, it does satisfy a weakening of it, the so-called stochastic triangle inequality. This guarantees that, analogously to Euclidean GIRGs, MCD-GIRGs have a large clustering coefficient~\cite{lengler2017existence}.

%\JL{We should discuss this definition a bit. Mention the expected degree statement, and that the actual degree is binomially distributed, which converges in the limit to the Poisson distribution. Say a few words about the assumption $\tau \in (2,3)$ (refer to a paper which estimates some values for real networks). Say a few words about $\alpha$. Mention the ball of influence, and that a constant fraction of vertices are in there. Mention the stochastic triangle inequality, why is it satisfied, and why it implies a large clustering coefficient. }

\subsection{Main Result}\label{sec:mainresult}
In order to state our main result, we first formally introduce expander graphs. We call a set $S$ of vertices $\zeta$-expanding if the external neighbourhood of $S$ is at least $\zeta$ times the size of $S$ itself. 

\begin{definition}[Expanding set]\label{def:expanders}
    Let $\zeta>0$. For a graph $\CG =(\CV, \CE)$, we call a subset $S \subset \CV$ of vertices a \emph{$\zeta$-expanding set (for $\CG$)} if the external neighbourhood $N_{ext}(S) = \{ v \in \CV \setminus S : \exists u \in S \textnormal{ with } uv \in \CE\}$ of $S$ is of size at least $|N_{ext}(S)| \ge \zeta \cdot \abs{S}$. We call $\zeta$ the \emph{expansion factor} of the expanding set.
\end{definition}

%\US{Add informal description (before) and discussion (after) of our main theorem}

We are now ready to formally state our main result, which asserts that in MCD-GIRGs, the subgraph induced by vertices of sufficiently large polylogarithmic weight is an expander with high probability. Furthermore, the expansion factor is typically of polylogarithmic order (and always at least constant).

\begin{restatable}{theorem}{MainTheorem}
    \label{thm:main-thm}
    For each $d \ge 2$ there is $c_d > 1$ such that the following holds. Let $\CG=(\CV,\CE)$ be an MCD-GIRG on $n$ vertices of dimension $d \ge 2$, power-law exponent $\tau\in(2,3)$ and $\alpha>1$, let $\gamma > \frac{1}{3-\tau}$ and $c'>0$ be constants and consider the subgraph $\CG'=(\CV',\CE')$ induced by vertices of weight at least $c'\log^{\gamma} n$. There exists a constant $\eps>0$ such that with high probability all subsets $S\subset \CV'$ satisfying $|S| \le \eps |\CV'|$ %\US{Is that a meaningful upper bound for $|S|$ ? should we push for any arbitrary constant fraction ? what comes out from the proof ?} 
    are expanding sets for $\CG'$ with expansion factor at least
    \begin{align}\label{eq:expansion-factor}
    \eps\cdot\min\Big\{\log^{\gamma(3-\tau)}n, \br{\tfrac{|\CV'|}{\aS}}^{1-1/c_d}\Big\}.
    \end{align}
    The same result holds if we replace $\CG'$ by the subgraph of $\CG$ induced by all vertices with weight in the interval $[c_1'\log^{\gamma} n,c_2'\log^{\gamma} n]$ for two arbitrary constants $c_2'>c_1'>0$. Moreover, all statements are still true if we replace $\CG'$ by the graph induced by all vertices of \emph{degree} (in $\CG$) at least $c'\log^\gamma n$, or by all vertices of degree in the interval $[c_1'\log^{\gamma} n,c_2'\log^{\gamma} n]$.
\end{restatable}

The larger the weights we consider, the better the expansion factor $\log^{\gamma(3-\tau)}n$ becomes. Eventually, when the size of $S$ approaches $\CV'$, this expansion factor is no longer achievable simply because we run out of vertices. In this case, we prove a smaller expansion factor, but this factor is still $\omega(1)$ whenever $|S| =o(|\CV'|)$ and becomes constant for $|S| = \Theta(|\CV'|)$. Note that this is optimal, simply because the number of remaining vertices in $\CV'\setminus S$ is then $O(|S|)$.

We highlight that the assumption on $d\ge2$ is crucial, as in one dimension, the MCD is equivalent to the Euclidean metric, where we encounter small separators at any scale. We also remark that each vertex in $\CV$ has probability $\Theta(\log^{\gamma(1-\tau)} n)$ to have weight at least $\log^\gamma n$. Therefore, whp $|\CV'| = \Theta(n\log^{\gamma(1-\tau)} n) = n^{1-o(1)}$.

Moreover, let us discuss where the condition $\gamma > \tfrac{1}{3-\tau}$ comes from. By~\cite[Lemma 10]{koch2021bootstrap} the expected degree of a vertex $v\in \CV'$ of weight $w_v$ in $\CG'$ is $\mathbb{E}[\deg_{\CG'}(v)] = \Theta(w_v\log^{\gamma(2-\tau)})$. Since $w_v \ge \log^\gamma n$, this is $\Omega(\log^{\gamma(3-\tau)}n) = \omega(\log n)$ as $\gamma > \tfrac{1}{3-\tau}$. Hence, every vertex has expected degree $\omega(\log n)$ in $\CG'$, and also in $\CG''$. In particular, whp $\CG'$ and $\CG''$ do not have isolated vertices. 

The condition on $\gamma$ is also tight, for the same reason. By the same calculation as above, for $\gamma < \tfrac{1}{3-\tau}$ the expected degree of a vertex of weight $\Theta(\log^\gamma n)$ in $\CG'$ or $\CG''$ is $o(\log n)$, and it is not hard to see that then the induced subgraph contains isolated vertices whp, which implies that it is not an expander. Indeed, since the degree of a vertex $v$ with given weight $w_v = \Theta(\log^\gamma n)$ in $\CG'$ is binomially distributed, any vertex of such a weight - of which there are $\Theta(n \cdot \log^{\gamma(1-\tau)}n) = n^{1-o(1)}$ many - would be isolated in $\CG'$ with probability at least $e^{-o(\log n)} = n^{-o(1)}$. So in expectation, there would be $\Theta(n^{1-o(1)} \cdot n^{-o(1)}) = n^{1-o(1)}$ isolated vertices in the induced subgraph. Those form a (large) set which is not expanding at all, as it has zero neighbours. This also holds whp using an Azuma-Hoeffding-type estimate.\footnote{See also Remark 5.13. in~\cite{jorritsma2024clustersizedecaysupercriticalkernelbased} which shows that whp there exists a vertex with weight less than $\log^{\frac{1}{3-\tau}}n$ which lies outside the giant component.} Hence, the condition $\gamma > \tfrac{1}{3-\tau}$ is tight.

Finally, we note that the case $\eps |\CV'| \le |S| \le \tfrac{1}{2}|\CV'|$ has already been treated in very similar form in~\cite{lengler2017existence}, which is why we exclude this case from the statement. There, it was shown that any linear-sized vertex subset of the giant component of $\CG$ has an external neighbourhood of linear size and hence constant expansion. This statement and its proof can readily be adapted to give the analogous result for the induced subgraph $\CG'$. Since $\CG'$ is connected, we may also drop the condition on the giant component, hence obtaining that every set $S\subseteq \CV'$ with $\eps |\CV'| \le |S| \le \tfrac{1}{2}|\CV'|$ has external neighbourhood of size at least $\eps|\CV'|$.

%Note that this also implies that the more weight-restrictive -- in terms of $\gamma$ -- the considered subgraph, the stronger the maximal expansion.

%Beyond the dimensionality requirement, we impose only two restrictions on the induced subgraphs: they consist only of vertices of high enough (at least $\log^{1/(3-\tau)}n$) weight and they must not be too large -- at most a small enough $\varepsilon$ fraction of all vertices ($n \log^{\gamma(1-\tau)}n$) above this weight cut-off. Requiring that the weights grow with increasing $n$ is natural, as there are many isolated vertices of constant weight. It is not difficult to see that the latter condition is also natural and, up to the size of $\varepsilon$, tight: given the chance to take at least some $1-o(1)$ fraction of all vertices above the weight cut-off, the adversary could just greedily select all vertices that are up for grabs, leaving only a $o(1)$ fraction as potential external neighbours. This would naturally result in an expansion of $o(1)$ as well. Even for small enough constant fractions of \emph{all} vertices above the weight-cutoff, the best expansion factor we can hope for is a small constant. In this case, the minimum in the theorem is attained by the second term, which indeed is a constant for $|S|=\Theta(n \log^{\gamma(1-\tau)}n)$. Note further that we do not assume anything about the component in which the vertices are located.

%% file: proofsketch.tex
% \MK{
% What we should include in the proof sketch is the following (if we understand correctly, so far most of the sketch is about how to count the external neighbors that we get). We play a two-player adversarial game (do not have to phrase it as a game but this is one way of making it very clear). 
% \begin{itemize}
%     \item We sample an MCD-GIRG and restrict to the subgraph induced by vertices of weight at least $\log^\gamma n$ and overlay it with the grid induced by the strips (of width $\frac{\log^{2\gamma}n}{n}$?)
%     \item The good player now claims: If the bad player has to choose $|S|$ many vertices, then whp, there will be at least one of the $d$ dimensions, such that projecting the coordinates of the chosen vertices intersect at least $k+1$ intervals (i.e. strips).
%     \item (optional) the bad player can then try and choose the $|S|$ many vertices and whp the good player's prediction will come true.
%     \item now given this outcome, by our (under-)counting method, as described in the subsequent, the vertex set $S$ will indeed have at least as many external neighbours as claimed (namely each of the $k+1$ strips will at least contain one vertex, which will have degree restricted to the strip at least $\log^{\gamma (3-\tau)}n$ and even if all the other vertices in $S$ were (over)counted as external neighbours (i.e. we can subtract $|S|$), we still have the claimed expansion factor.
% \end{itemize}}
%We begin this section with an intuitive proof overview, followed by the detailed proof. 
In this section, we provide an intuitive overview of the proof. The formal, detailed proof can be found in Appendix~\ref{sec:proof_appendix}.

We want to lower-bound the external neighbourhood of a given vertex set $S$. %To show this, we consider the following adversarial game. %Note that we choose this framing to make the proof more tangible but the statements we prove are whp non-existence results and not restricted to computational intractibility.
We begin by sampling an MCD-GIRG and restricting to the subgraph induced by vertices of weight at least $\log^{\gamma}n$. Next, we partition the ground space into \emph{strips} of width $\ell=\Theta(\tfrac{\log^{2\gamma}n}{n})$ along each dimension. Figure \ref{fig:illustrations} (right) illustrates this. The idea behind this strip width is that two vertices of weight $\log^{\gamma}n$ in the same strip have a connection probability in $\Theta(1)$.

%Now, we select a (set) size $s$ and an integer $k$. The adversary then selects a set of $s$ vertices and tries to minimize the number of strips which the set intersects in each dimension. 
For each (set) size $s$ we will find a suitable value $k$ such that whp the following claim holds: 

\emph{``There is no set $S$ of vertices with $\abs{S} \ge s$ such that in every dimension, $S$ intersects at most $k$ strips.''}

In a second step, we show that every vertex $v\in \CV'$ has whp $\Omega(w_v\log^{\gamma(2-\tau)}n) = \Omega(\log^{\gamma(3-\tau)} n)$ neighbours \emph{in the same strip}, in both $\CG'$ and $\CG''$. Since every set $S$ of size $s$ intersects at least $k$ strips in some dimension, the vertices in $S$ have at least $k\cdot \Omega(\log^{\gamma(3-\tau)}n)$ distinct neighbours. 
%
%For suitably chosen $k$, as we will show, the first player's claim holds with high probability. 
%Given this claim, the following (under-)counting argument ensures that the vertex set $S$ has at least as many external neighbours as claimed in Theorem \ref{thm:main-thm}. The claim implies that any set $S$ of size $\abs{S} \ge s$, intersects at least $k+1$ strips in at least one dimension $i$. Each of these $k+1$ strips will contain at least one vertex from $S$. For each strip in dimension $i$ we select (at most) one representative vertex from $S$ and count their neighbours, considering only those within the same strip in dimension $i$.
%Summing these neighbourhoods provides a lower bound on the total neighbourhood size.
To obtain the size of the \emph{external} neighbourhood, we must ensure that vertices inside $S$ itself are not counted. % as their own neighbours. 
By simply subtracting $\abs{S}$ from the total neighbourhood count, we obtain the desired lower bound on the size of the external neighbourhood.
\newcommand{\drawTriangle}[3]{%
    \fill[black] (#1,#2) -- ++(#3/2,0) -- ++(-#3/2,#3) -- ++(-#3/2,-#3) -- cycle;
}
\vspace{12pt}
\begin{figure}[ht]
\label{fig:balls}
\centering
\begin{tikzpicture}[scale=4]
    
    % First figure (left side)
    \begin{scope}[shift={(0,0)}] % Shift
        \fill[blue!50, opacity=0.5] (0.1,0) rectangle (0.3,1);
        \fill[blue!50, opacity=0.5] (0,0.1) rectangle (0.1,0.3);
        \fill[blue!50, opacity=0.5] (0.3,0.1) rectangle (1,0.3);
        
        \fill[red!50, opacity=0.5] (0.7,0) rectangle (0.9,1);
        \fill[red!50, opacity=0.5] (0,0.7) rectangle (0.7,0.9);
        \fill[red!50, opacity=0.5] (0.9,0.7) rectangle (1,0.9);

        \draw[thick] (0,0) rectangle (1,1);
        
        % Triangles
        \drawTriangle{0.2}{0.18}{0.05}
        \drawTriangle{0.8}{0.78}{0.05}
    \end{scope}

    % Second figure (right side)
    \begin{scope}[shift={(1.5,0)}] % Adjust position if needed
        \fill[red!50, opacity=0.5] (0.2,0) rectangle (0.4,1);
        
        % Horizontal strips
        \fill[red!50, opacity=0.5] (0,0.4) rectangle (0.2,0.6);
        \fill[red!50, opacity=0.5] (0.4,0.4) rectangle (0.8,0.6);
        \fill[red!50, opacity=0.5] (0.8,0.4) rectangle (1,0.6);
        \fill[red!50, opacity=0.5] (0,0.6) rectangle (0.2,0.8);
        \fill[red!50, opacity=0.5] (0.4,0.6) rectangle (0.8,0.8);
        \fill[red!50, opacity=0.5] (0.8,0.6) rectangle (1,0.8);
        
        \draw[thick] (0,0) rectangle (1,1);
        \draw[step=0.2, gray, very thin] (0,0) grid (1,1);
        
        % Points
        \foreach \x/\y in {0.25/0.25, 0.10/0.35, 0.12/0.09, 0.17/0.85, 0.23/0.88, 0.47/0.35, 0.52/0.16, 0.63/0.89, 0.64/0.37, 0.72/0.18, 0.81/0.53, 0.85/0.25,0.07/0.50, 0.09/0.65, 0.56/0.62, 0.93/0.74} {
            \fill[black] (\x,\y) circle[radius=0.015];
        }
        
        % Triangles
        \drawTriangle{0.33}{0.43}{0.05}
        \drawTriangle{0.24}{0.67}{0.05}
        \drawTriangle{0.34}{0.72}{0.05}
        
        % Labels
        \draw[decorate, decoration={brace, amplitude=5pt, mirror}, thick] (1.03,0.2) -- (1.03,0.4)
            node[midway, right, xshift=+5pt] {strip};
        %\node at (0.5,-0.1) {Dimension $1$};
        %\node[rotate=90] at (-0.1,0.5) {Dimension $2$};
    \end{scope}
\end{tikzpicture}
\caption{\textbf{Left:} Balls in the Minimum-Component Distance are cross-shaped.
\textbf{Right:} Illustration of the proof technique in two dimensions. The vertices marked by triangles form the adversarially chosen set $S$ for $s = 3$. Although the adversary has managed to have $S$ only intersect with one vertical strip, it does not achieve $k=1$ because $S$ intersects with two horizontal strips. In fact, for $s=3$ and $k= 1$ the claim holds as there is no set of three vertices spanning only one strip in each dimension.}
\label{fig:illustrations}
\end{figure}
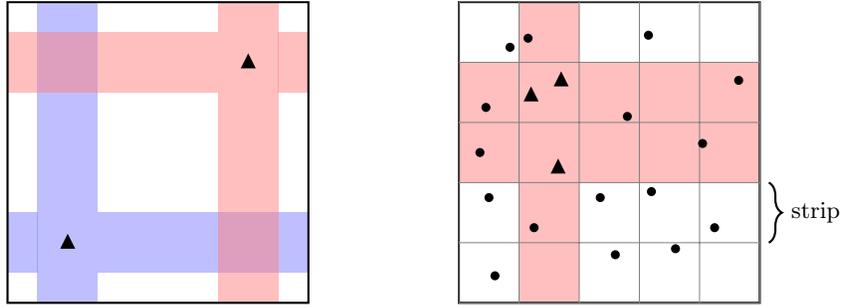
\vspace{-12pt}
%We note that in order to determine the size of the external neighbourhood precisely, we would need to correct for the over-counting of vertices which belong to the -- intersecting -- neighbourhoods of multiple vertices in $S$ (left panel in Figure~\ref{fig:illustrations}). However, it will be sufficient for our lower bound to restrict our attention, for each vertex in $S$, to the proportion of neighbours which are close along one of the coordinate axes, which by definition of the Minimum-Component Distance is at least roughly a $\frac{1}{d}$-fraction of its total number of neighbours.

Formalizing the strategy outlined above, our main result will leverage the following two key propositions, which are proven in Appendix \ref{sec:proof_appendix} along with the theorem itself. The first proposition provides a lower bound on the number of strips\footnote{For the formal definition of $i$-strips see Definition~\ref{def:strip}.} which a set of at least $s$ many vertices will intersect whp. 
\begin{restatable}{proposition}{stripcovering}\label{prop:strip-covering}
    For each $d \ge 2$ there exists $c_d>1$ such that the following holds. Let $\CG=(\CV,\CE)$ be an MCD-GIRG on $n$ vertices of dimension $d \ge 2$ with power-law exponent $\tau\in(2,3)$ and $\alpha>1$, let $\gamma > \frac{1}{3-\tau}$ and $c'>0$ be consatnts and consider the subgraph $\CG'=(\CV',\CE')$ induced by vertices of weight at least $c'\log^{\gamma} n$. For every $s=\omega(1)$ there exists
    \begin{align}\label{eq:k-condition}
        k = s \cdot \Omega\Big(\min\Big\{1, \tfrac{1}{\log^{\gamma(3-\tau)}n} \cdot \br{\tfrac{|\CV'|}{s}}^{1-1/c_d}\Big\}\Big),
    \end{align}
    such that the following holds with high probability: for every $S\subset \CV'$ with $|S| \ge s$, there is a coordinate $1\le i \le d$ such that $S$ has non-empty intersection with at least $k$ distinct $i$-strips.
    \end{restatable}
The proof of Proposition~\ref{prop:strip-covering} is a counting argument. We simply count the number of sets of $k$ strips in each dimension and the number of sets $S$ of $s$ vertices, and estimate the probability that $S$ is contained in the chosen strips. By a union bound over all sets of strips and all $S$ we show that whp there exists no combination of $k$ strips per dimension and set $S$ that falls entirely into these strips.

The second proposition gives a lower bound on the number of neighbours of weight $\Theta(\log^\gamma n)$ of any set $S\subseteq \CV'$, expressed as a function of $k$.
% %$k$ that are covered in at least one dimension. Here, \emph{covered} means that each selected strip contains at least one vertex from $S$.

\begin{restatable}{proposition}{inclusivenbhd}
\label{prop:inclusive_nbhd}
    %For each $d \ge 2$ there exists $c_d>1$ such that the following holds.
    Let $\CG=(\CV,\CE)$ be an MCD-GIRG on $n$ vertices of dimension $d \ge 2$ with power-law exponent $\tau\in(2,3)$ and $\alpha>1$, let $\gamma > \frac{1}{3-\tau}$ and $c'>0$ $c_2'> c_1'>0$ be constants and consider the subgraph $\CG'=(\CV',\CE')$ induced by vertices of weight at least $c'\log^{\gamma} n$. Then with high probability the following holds: for every $S \subset \CV'$ that has non-empty intersection with at least $k$ distinct $i$-strips (for some coordinate $1\le i \le d$), the number of neighbours $N_{c_1',c_2'}(S)$ with weights in the interval $[c_1'\log^\gamma n,c_2'\log^\gamma n]$ of $S$ is at least 
    \begin{align*}
        N_{c_1',c_2'}(S) \ge \Omega(k \cdot \log^{\gamma (3-\tau)}n).
    \end{align*}
\end{restatable}
As outlined above, we prove the proposition by showing that each vertex in $\CV'$ has in expectation $\Omega(\log^{\gamma (3-\tau)}n)$ neighbours in the same strip. Since the expectation is $\omega(\log n)$, the actual number of neighbours is concentrated, and we can afford a union bound to show that whp \emph{every} vertex has $\Omega(\log^{\gamma (3-\tau)}n)$ neighbours in the same strip. 
%The same statement is also true for $\CG''$. 
We refer the reader to Figure~\ref{fig:sample_mcd_girg}, which displays how the edges are in approximate alignment with the strips. 

Taken together, these statements imply the theorem for $s=\omega(1)$. Constant-sized sets require a separate but analogous argument.

%% file: tools.tex
%We present 

%\section{Preliminaries}\label{sec:preliminaries}

In this short section, we provide two technical tools that will come in handy during the proof, starting with the classical Chernoff bound.
\begin{lemma}[Chernoff bound]
\label{lem:Chernoff}
    Let $X = \sum_{i= 1}^{n}{X_i}$ be the sum of independent indicator random variables $X_i$. Then for any $\eps \in (0,1)$ we have
    \begin{align*}
        \pr \bre{X \ge (1+\eps) \Expected{X}} &\leq \exp\br{-\eps^2 \Expected{X}/3},\\
        \pr \bre{X \leq (1-\eps) \Expected{X}} &\leq \exp\br{-\eps^2 \Expected{X}/2}.
    \end{align*}
    
\end{lemma}

We conclude the section with the following useful observation which allows us to bound logarithms of binomial coefficients.
%\US{I think we should change the phrasing of this lemma}
\begin{lemma}[Stirling]
\label{lem:stirling}
   %Let $R(N) \coloneqq \tfrac{1}{2} \log\br{2\pi N} + o(1)$. Then we 
   Let $R(N) \coloneqq \tfrac{1}{2} \log\br{2\pi N} $. Then we 
   can express the logarithm of the binomial coefficient $\binom{N}{k}$ as
    \begin{align*}
        \log{\binom{N}{k}} \le \br{N-k}\log{\br{\frac{N}{N-k}}} + k\log{\br{\frac{N}{k}}} + R(N) - R(N-k) - R(k).
    \end{align*}
\end{lemma}
\begin{proof}
 The statement follows immediately from Stirling's formula, which yields that $\log{\br{N!}} = N \log N - N +R(N)$.%, where $R(x) = 0.5 \log{\br{2\pi x}} + o(1)$.
\end{proof}
%\KS{Does it even make sense to have a tools section, if it only contains these two lemmata?}
%\US{good point, maybe it could be a subsection of section \ref{sec:mcd_proof}}
%\MK{or we include it in the previous section (notation, tools and preliminaries)}
%\MK{By the way, what is $R_i$ referring to?}
%\KS{Was a leftover from better times ;)} \MK{Haha okay:)}

\begin{lemma}\label{lem:degree-concentration}
    Let $\CG=(\CV,\CE)$ be an MCD-GIRG on $n$ vertices in dimension $d\ge 1$ with power-law exponent $\tau\in(2,3)$ and $\alpha>1$, let $\gamma > \frac{1}{3-\tau}$ and consider the subgraph $\CG'=(\CV',\CE')$ induced by vertices of weight at least $c\log^{\gamma} n$ for some constant $c>0$. Then whp 
    \begin{align*}
    |V'| = \Theta(n\log^{\gamma(1-\tau)} n) = n^{1-o(1)}.
    \end{align*}
    Moreover, there is a function $h=o(1)$ such that whp all $v\in \CV'$ satisfy
    \begin{align}\label{eq:degree-concentration}
    \begin{split}
        \deg_{\CG}(v) & \in (1\pm h)\EE{\deg_{\CG}(v) \mid w_v} \text{ and} \\
        \deg_{\CG'}(v) & \in (1\pm h)\EE{\deg_{\CG'}(v) \mid w_v},
    \end{split}
    \end{align}
    where we use the notation $1\pm h$ as a shortcut for the interval $[1-h, 1+h]$. Finally, there is $c>0$ such that whp all vertices $v\in \CV$ of weight $w_v \le c \log^{\gamma}n$ have degree $\deg_{\CG}(v) < \log^{\gamma} n$.
\end{lemma}
\begin{proof}
    Each vertex in $\CV$ has probability $\Theta(\log^{\gamma(1-\tau)} n)$ to have weight at least $c\log^\gamma n$, independently of each other. Therefore, whp $|\CV'| = \Theta(n\log^{\gamma(1-\tau)} n)$. Moreover, for $v\in \CV'$ with given weight $w_v$ the degree of $v$ in either $\CG$ or $\CG'$ follows a binomial distribution by the symmetric definition of the model. The expectation is $\Theta(w_v)$ in $\CG$ and $\Theta(w_v\log^{\gamma(2-\tau)}n)$ in $\CG'$ by~\cite[Lemma 10]{koch2021bootstrap}. Since $w_v \ge c\log^\gamma n$, the latter expectation is $\Omega(\log^{\gamma(3-\tau)}n) = \omega(\log n)$, and the former expectation is also $\omega(\log n)$. Hence,~\eqref{eq:degree-concentration} holds with probability $1-e^{-\omega(\log n)} = 1-o(1/n)$ for $v$ by the Chernoff bounds. A union bound over all $v\in \CV'$ implies that whp~\eqref{eq:degree-concentration} holds for all $v\in \CV'$.

    For the last statement, if $c>0$ is sufficiently small then the expected degree of any vertex $v$ of weight $w_v \le c \log^{\gamma}n$ is at most $\tfrac12 \log^\gamma n$. Since the degree is binomially distributed, $\Pr{\deg(v) \ge \log^\gamma n \mid w_v} \le e^{-\Omega(\log^\gamma n)} = o(1/n)$ by the Chernoff bound, and the statement follows by a union bound over all $v$.
\end{proof}

%% file: proof.tex
Turning to the formal proof of our main theorem, we begin by making precise the concepts of strips.
\begin{definition}
\label{def:strip}
    Consider the
    %an MCD-GIRG $\CG=(\CV,\CE)$ on the 
    torus $\T^d$ and let $\ell\in(0,1)$ be such that $\ell^{-1}\in\N$. For $1\le i \le d$ we define an \emph{$i$-strip of width $\ell$} to be a subset of $\T^d$ of the form $\{(x_1, \ldots, x_d)\in\T^d : x_i \in [j\cdot \ell, (j+1)\cdot \ell)\}$ where $0 \le j < \ell^{-1}$ is an integer. 
\end{definition}
Throughout the proof (and the paper as a whole), we will partition the $i$-th dimension into strips of width $\ell\coloneqq \lfloor\tfrac{n}{\log^{2\gamma}n}\rfloor^{-1}$, and simply refer to them as \emph{strips} or \emph{$i$-strips} as their width is clear from context. For a vertex $v\in \CV$ with position $x_v$, we denote by $I_{i,v}$ the (unique) $i$-strip that contains $v$ (i.e.\ that satisfies $(x_v)_i\in I_{i,v}$). The idea behind this strip width is that two vertices of weight $\log^{\gamma}n$ in the same strip have a connection probability in $\Theta(1)$. This will be shown in Proposition \ref{prop:inclusive_nbhd}.
%\KS{Does it make sense to already formalize the concepts of strips before the proof sketch?}
%\US{I think it makes more sense to have an informal proof sketch (where we say informally what a strip is), and afterwards have all the formal definitions and statements and proofs}

The first proposition below is at the heart of the proof. It gives a lower bound on the number of strips which a set of at least $s$ many vertices will intersect whp. 
% \begin{proposition}\label{prop:strip-covering}
%     For each $d\ge 2$ there exist $c_d >1$ such that the following holds. Let $\CG=(\CV,\CE)$ be an MCD-GIRG on $n$ vertices of dimension $d \ge 2$ with power-law exponent $\tau\in(2,3)$, let $\gamma > \frac{1}{3-\tau}$ and consider the subgraph $\CG'=(\CV',\CE')$ induced by vertices of weight at least $\log^{\gamma} n$. For every $s=\omega(1)$ there exist
%     \begin{align}\label{eq:k-condition}
%         k = s \cdot \Omega\br{\min\left\{1, \frac{1}{\log^{\gamma(3-\tau)}n} \cdot \br{\frac{|\CV'|}{s}}^{1-1/c_d}\right\}},
%     \end{align}
%     such that the following holds with high probability: for every $S\subset \CV'$ with $|S| \ge s$, there is a coordinate $1\le i \le d$ such that $S$ has non-empty intersection with at least $k$ disjoint $i$-strips.
% \end{proposition}
\stripcovering*

\begin{proof}
We begin the proof with a remark that justifies some simplifying assumptions.
\begin{remark}\label{rem:bounds-on-k}
    Clearly we can assume that $k \le s$. Moreover, by Chernoff bounds and Definitions~\ref{def:power-law} and~\ref{def:mcd_girgs}, we have $|\CV'|=\Theta(n \log^{\gamma(1-\tau)}n)$ whp. Indeed, the expected number of vertices of at least some weight $w$ is of order $\Theta(n \cdot w^{1-\tau})$, and so is the expected number of vertices of weight in the interval $[c_1w,c_2w]$ for any constants $c_2>c_1>0$. Since they are sampled independently and their expectation is clearly $\omega(1)$, the result also holds whp by Chernoff's inequality. Since $s \le |\CV'|$ we have
    \begin{align*}
        s \cdot \br{\frac{1}{\log^{\gamma(3-\tau)}n} \cdot \br{\frac{|\CV'|}{s}}^{1-1/c_d}} 
        = \frac{s^{1/{c_d}}|\CV'|^{1-1/{c_d}}}{\log^{\gamma(3-\tau)}n}
        \le \frac{|\CV'|}{\log^{\gamma(3-\tau)}n} = \Theta\br{\frac{n}{\ltg}},
    \end{align*}
    and hence we can also assume that $k \le \lfloor n/\ltg \rfloor$.
\end{remark}

    Now, we actually prove the following. Given $s\in\N$, we show that whp an adversary cannot find a set $S$ of size $\aS = s$ along with $k$ $i$-strips in each dimension $1\le i \le d$ such that $S$ is contained in the union of these $k$ $i$-strips for all $i$. %By completely contained, we mean that for every $1 \leq i \leq d$, the set of $i$-th vertex coordinates $\{(x_v)_i : v \in S\}$ falls entirely within the union of the $k$ selected strips.
    Notice that this implies the desired result.

%    We start by drawing the weights of all vertices, which in particular determines the set $\CV'$, and denote by $N \coloneqq |\CV'|$ the number of vertices of weight at least $\log^{\gamma} n$. Note that by Chernoff $N=\Theta(n \log^{\gamma(1-\tau)}n)$ whp. For the rest of the proof we condition on $N$, and assume in particular that $N=\Theta(n \log^{\gamma(1-\tau)}n)$.

    Let $p_{k,s}$ denote the probability that there exists a set $S\subset \CV'$ of size $s$ as well as $k$ $i$-strips for each dimension $1\le i \le d$ that satisfy the above property. Note that the number of choices for $S$ is $\binom{|\CV'|}{s}$ and the number of choices for the strips is $\binom{\lfloor n/\ltg\rfloor}{k}^d$. For a fixed subset $S\subset\CV'$ of size $s$ as well as $k$ fixed $i$-strips, the probability that $S$ is contained in the union of these $i$-strips is $(k\lfloor n/\ltg\rfloor^{-1})^s$ since the volume of the union of the $i$-strips is $k\lfloor n/\ltg\rfloor^{-1}$. Since the coordinates of the vertices are independently sampled, a union bound yields the upper bound
%
    %Assume for each vertex $v$, the weight $w_v$ is already chosen but its position is not yet revealed. Let the adversary choose an arbitrary combination of sets and strips. %set-strip combination. The probability that all selected vertices fall in the union of the chosen strips in each dimension is $\br{k\log^{2\gamma}n/n}^{ds}$. 
    %Applying a union bound over all possible set-strip combinations provides an upper bound on the probability that an adversary can successfully find such a combination.
%
    %Let $N$ denote the total number of vertices of weight at least $\log^{\gamma}n$. As there are $n/\log^{2\gamma}n$ strips in each of the $d$ dimensions, the total number of possible set-strip combinations is given by $\binom{N}{s} \cdot \binom{n/\ltg}{k}^d$. All together, the probability the adversary can find a suitable set-strip combination is at most 
    \begin{align*}
        p_{k,s} \le \binom{|\CV'|}{s} \cdot \binom{\lfloor n/\ltg\rfloor}{k}^d \cdot \left(\frac{k}{\lfloor n/\ltg\rfloor}\right)^{ds}.
    \end{align*}
    Using Lemma \ref{lem:stirling}, and remembering that $R(N) = \tfrac{1}{2} \log\br{2\pi N} + o(1)$, we compute
    \begin{align*}
        &p_{k,s} \le \exp \bre{\log{\binom{|\CV'|}{s}} + d\log {\binom{\lfloor n/\ltg\rfloor}{k}} + ds \log{\br{\frac{k}{\lfloor n/\ltg\rfloor}}}}\\
        &= \exp \Biggl[\underbrace{\br{s - |\CV'|} \log{\br{1 - \frac{s}{|\CV'|}}}}_{\eqqcolon B_1} 
        + s \log{\br{\frac{|\CV'|}{s}}} + \underbrace{R(|\CV'|) - R(|\CV'|-s) - R(s)}_{\eqqcolon R_1} \\
        &\quad\quad\underbrace{-d\br{\lfloor n/\ltg\rfloor - k} \log{\br{1 - \frac{k}{\lfloor n/\ltg\rfloor}}}}_{\eqqcolon B_2} 
        + dk \log{\br{\frac{\lfloor n/\ltg\rfloor}{k}}} \\
        &\quad\quad + \underbrace{d (R(\lfloor n/\ltg\rfloor)-R(\lfloor n/\ltg\rfloor-k) -R(k))}_{\eqqcolon R_2} + ds \log{\br{\frac{k}{\lfloor n/\ltg\rfloor}}}  \Biggr]\\
        &= \exp \Biggl[-s\br{d\br{1 - \frac{k}{s}} \log{\br{\frac{\lfloor n/\ltg\rfloor}{k}}} - \log\br{\frac{|\CV'|}{s}}} + B_1 + B_2 + R_1 + R_2 \Biggr],
    \end{align*}
    %where $R_1$ and $R_2$ are the two remainder terms of the Stirling Approximation we used.
    where for the last inequality we used the fact that $|\CV'| \le n$ and $\lfloor n/\ltg\rfloor \le n$. %We treat the last components of the exponential function separately, starting with $B_1$. 
    
    We start by examining $B_1$. If $s = |\CV'|$, then $B_1 = 0$. For the case $s < |\CV'|$, we make use of the fact that $\log\br{1+x} \ge x/(1+x)$ for $x>-1$ to get:
    \begin{align*}
        B_1 
        = s\br{1 -\frac{|\CV'|}{s}} \cdot \log{\br{1- \frac{s}{|\CV'|}}} 
        \leq s\br{1 -\frac{|\CV'|}{s}} \cdot \frac{-s}{|\CV'|-s}
        = s.
    \end{align*}
    
    Moving on to $B_2$, recall that $k \le s$ and $k \le \lfloor n / \ltg \rfloor$ by Remark \ref{rem:bounds-on-k}.
    %recall that the torus $\T^d$ is divided into $\lfloor n/\ltg \rfloor$ strips in each dimension. Hence this is an upper bound for $k$. Additionally, notice that $k \leq s$ (otherwise the statement would trivially be false, as $s$ vertices can intersect at most $k$ strips in each dimension).
    Therefore, using again the inequality $\log\br{1+x} \ge x/(1+x)$, we have
    \begin{align*}
        B_2 
        = - d\br{\lfloor n / \ltg \rfloor-k} \log{\br{1 - \frac{k}{\lfloor n / \ltg \rfloor}}} 
        % \leq d\br{\frac{n}{\ltg}-k} \frac{k \log^{2\gamma}n}{n- k \log^{2\gamma}n}
        \leq dk \le ds.
    \end{align*}
    
%    Turning to $dR_1$ and $R_2$, we get
%    \begin{align*}
%        d R_1 + R_2 = \frac{1}{2}\br{d\log\br{\frac{n}{2\pi k \br{n-k \ltg}}} + \log\br{\frac{|\CV'|}{2\pi s \br{|\CV'| - s}}}} + o(1) = o\br{s}.
%    \end{align*}
%    These terms are thus all bounded by $c'\cdot s$ for some constant $c'$. By Definition \ref{def:power-law} and by Chernoff's Bound, whp there exists a constant $c$, such that $|\CV'|= c n \log^{\gamma(1-\tau)}n$. Leveraging this and bringing everything back together, we get 

    Using the identity $\log a - \log b = \log \tfrac{a}{b}$, by plugging in the terms and elementary algebraic manipulations, we easily get that $R_1 + R_2 = O(1)$. Therefore, $B_1 + B_2 + R_1 + R_2 \le 2ds$ for sufficiently large $s$, and we deduce that
    \begin{align*}
        p_{k,s} &\le \exp \Biggl[-s\br{d\br{1 - \frac{k}{s}} \log\br{\frac{\lfloor n/\ltg\rfloor}{k}} - \log\br{\frac{|\CV'|}{s}} - 2d}\Biggr]\\
        &= \exp \bre{-s\br{d\br{1 - \frac{k}{s}} \log\br{\frac{s \lfloor n/\ltg\rfloor}{k |\CV'|} \br{\frac{|\CV'|}{s}}^{1-1/(d(1-\frac{k}{s}))}}-2d}}.
    \end{align*}
    Note that, using Chernoff's bound and Definitions~\ref{def:power-law} and~\ref{def:mcd_girgs}, whp we have $|\CV'|= \Theta( n \log^{\gamma(1-\tau)}n)$, and in particular there exists a constant $c>0$ such that $\frac{\lfloor n/\ltg\rfloor}{|\CV'|} \ge \frac{1}{c\log^{\gamma(3-\tau)}n}$. We condition on this for the rest of the proof, and obtain
    \begin{align*}
        p_{k,s} \le \exp \bre{-s\br{d\br{1 - \frac{k}{s}} \log\br{\frac{s}{k c \log^{\gamma(3-\tau)}n} \br{\frac{|\CV'|}{s}}^{1-1/(d(1-\frac{k}{s}))}}-2d}}.
    \end{align*}
    
    It remains to show that for some $k$ satisfying \eqref{eq:k-condition}, we have $p_{k,s}=o(1)$. Let
    \begin{align*}
        k \coloneqq s \cdot \min\left\{c_1, \frac{c_2}{\log^{\gamma(3-\tau)}n} \cdot \br{\frac{|\CV'|}{s}}^{1-1/c_d}\right\},
    \end{align*}
    for some positive constants $c_1, c_2 < 1$ and $c_d > 1$ to be specified later. Clearly this satisfies \eqref{eq:k-condition}. We split into two cases depending on where the minimum is taken in the above expression. 
    
    For the first case, assume that the first term is taken as the minimum, i.e.\ that
    \begin{align}\label{eq:min-first-term}
        \frac{k}{s}=c_1 \le \frac{c_2}{\log^{\gamma(3-\tau)}n} \cdot \br{\frac{|\CV'|}{s}}^{1-1/c_d}.
    \end{align}
    We get
    \begin{align*}
        p_{k,s} &\le \exp \bre{-s\br{d\br{1 - c_1} \log\br{\frac{1}{c_1 c \log^{\gamma(3-\tau)}n} \br{\frac{|\CV'|}{s}}^{1-1/(d(1-c_1))}}-2d}}.
    \end{align*}
    To establish that the above expression is $o(1)$, we choose $c_1<1$ and $c_d > 1$ small enough so that $c_d < d(1-c_1)$, which is possible since $d \ge 2$. Using \eqref{eq:min-first-term} we obtain
    \begin{align*}
        p_{k,s} &\leq \exp \bre{-s\br{d\br{1 - c_1} \log\br{\frac{1}{c_2 c} \cdot \br{\frac{|\CV'|}{s}}^{1/c_d-1/(d(1-c_1))} }- 2d}},
    \end{align*}
    which is $o(1)$ if we choose $c_2>0$ so small that $(1-c_1)\log((c_2c)^{-1}) > 2$ since $s = \omega(1)$.
    
    If, considering now the second case, the minimum is attained by the second term, i.e.\ if
    \begin{align}\label{eq:min-second-term}
        \frac{k}{s} = \frac{c_2}{\log^{\gamma(3-\tau)}n} \cdot \br{\frac{|\CV'|}{s}}^{1-1/c_d} \le c_1,
    \end{align}
    we get
    \begin{align*}
        p_{k,s} &\le \exp \bre{-s\br{d\br{1 - \frac{k}{s}} \log\br{\frac{1}{c_2 c} \cdot \br{\frac{|\CV'|}{s}}^{1/c_d-1/(d(1-\frac{k}{s}))}}-2d}}\\
        &\leq \exp \bre{-s\br{d\br{1 - c_1} \log\br{\frac{1}{c_2 c} \cdot \br{\frac{|\CV'|}{s}}^{1/c_d-1/(d(1-c_1))}}-2d}},
    \end{align*}
    where we used \eqref{eq:min-second-term} for the second inequality. Using the same reasoning as above shows that the values for the constants chosen above also guarantees that the expression becomes $o(1)$.
    
    %Using the same reasoning as above one can find values of $c_1, c_2 < 1$ and $c_d > 1$ for which this expression is $o(1)$, which concludes the proof.
$\hfill\qed$
\end{proof}

The next proposition gives a lower bound on the number of neighbours of $S$ with with weight $\Theta(\log^\gamma n)$, expressed as a function of the number of strips it intersects along some dimension.
%$k$ that are covered in at least one dimension. Here, \emph{covered} means that each selected strip contains at least one vertex from $S$.

% \begin{proposition}
% \label{prop:inclusive_nbhd}
%     Let $\CG=(\CV,\CE)$ be an MCD-GIRG on $n$ vertices of dimension $d \ge 2$, let $\gamma > \frac{1}{3-\tau}$ and consider the subgraph $\CG'=(\CV',\CE')$ induced by vertices of weight at least $\log^{\gamma} n$. Then with high probability the following holds: for every $S \subset \CV'$ that has non-empty intersection with at least $k$ disjoint $i$-strips (for some coordinate $1\le i \le d$), its neighbourhood has size at least 
%     \begin{align*}
%         \abs{\Gamma_{\CG'}(S)} \ge \Omega(k \cdot \log^{\gamma (3-\tau)}n).
%     \end{align*}
% \end{proposition}
\inclusivenbhd*
\begin{proof}
    To circumvent over-counting vertices in the neighbourhood of the set $S$, we consider at most a single vertex in $S$ for each $i$-strip and only look at its neighbourhood restricted to this $i$-strip. By showing that for a single such vertex its strip-restricted neighbourhood contains $\Omega(\log^{\gamma(3-\tau)}n)$ vertices of the right weight, the proposition then follows by summing up the neighbourhoods for the $k$ strips. We highlight that the fact that  each vertex in $\mathcal{V}'$ has enough neighbours in its own $i$-strip, for each coordinate $i$ is a property of the whole graph. This global property \emph{then} implies the property that each subset $S\subset\mathcal{V}'$ has a large enough neighbourhood. This is why we will not need to take a union bound over all subsets $S\subset\mathcal{V}'$ - but only over all vertices in $\mathcal{V}'$, which yields the global property. 

    Let $v \in S$ be an arbitrary vertex, remember that $I_{i,v}$ denotes the $i$-strip containing $v$, and consider another, arbitrary vertex $u\in  I_{i,v}\setminus \{v\}$ of weight $w_u \in I_W := [c_1'\log^\gamma n, c_2'\log^\gamma n]$ in this strip. Using Equation~\eqref{eq:connection}, we can lower bound the probability that the edge $uv$ exists (in $\CG'$) by
    \begin{align*}
        \pr[uv\in\CE & \mid x_{u,i} \in I_{i,v},w_u\in I_W,w_v\ge c'\log^\gamma n] \\
        %&= \Theta \left(\min\left\{1, \left(\frac{w_u w_v}{n \cdot V_{min}(\norm{x_u -x_v}_{min})}\right)^{\alpha} \right\}\right)\\
        & \ge \Omega\left(\min\left\{1, \left(\frac{\log^{2\gamma}n}{n \cdot \lfloor n/\ltg \rfloor^{-1}}\right)^{\alpha} \right\}\right)
        = \Omega(1),
    \end{align*}
    where we used that the width of a strip is $\lfloor n/\ltg \rfloor^{-1}$ and that the volume function scales as $V_{min}(r) = \Theta(r)$ in MCD-GIRGs. By Definitions~\ref{def:power-law} and~\ref{def:mcd_girgs}, the expected number of vertices $u$ in $\CV' \cap I_{i,v}$ with weight in $I_W$ is $\Theta(\log^{\gamma(3-\tau)}n)$, and hence the expected number of neighbours of $v$ with those properties is also $\Theta(\log^{\gamma(3-\tau)}n)$. By the Chernoff bound we deduce that $v$ has $\Theta(\log^{\gamma(3-\tau)}n)$ in $I_{i,v}$ with weight in $I_W$ with probability $1-\exp(-\Omega(\log^{\gamma(3-\tau)}n))=1-n^{-\omega(1)}$ (since $\gamma>\tfrac{1}{3-\tau}$). 
    %\begin{align*}
    %    \Expected{\deg_{\CV' \cap I_{i,v}}(v)}
    %    = \sum_{u \neq v \in \CV' \cap I_{i,v}}{\pr[u \sim v \land x_{u, i} \in I_{i,v}]}
    %    % &= \sum_{u \neq v}{\pr[u \sim v \mid x_{u,i} \in I_v] \cdot \pr[x_{u, i} \in I_v]}\\
    %    &= \sum_{u \neq v}{\Theta(1) \cdot \frac{\log^{2\gamma}n}n}
    %    = \Theta\br{\log^{\gamma(3-\tau)}n}.
    %\end{align*}
%    To obtain concentration, we use standard Chernoff bounds. We define $\mu = \Expected{\deg_{I_v}(v)}$ and $\eps = \log{n}/\sqrt{\mu} = o(1)$. Therefore 
%    \begin{align*}
%        \pr[\abs{\deg_{I_v}(v) - \mu} > \eps \cdot \mu] \leq e^{-\Theta(\eps^2 \cdot \mu)} = n^{- \omega(1)}.
%    \end{align*}
    The proposition follows by taking a union bound over all vertices in $\CV'$.
$\hfill\qed$
\end{proof}
% Having established the size of the neighbourhood, the following corollary provides an upper bound on the size of the external neighbourhood.
% \US{I would remove Proposition \ref{prop:exclusive_nbhd} as an intermediate result, and instead include the (very short) reasoning its proof contains directly inside the proof of Theorem \ref{thm:main-thm}}
% \begin{proposition}
% \label{prop:exclusive_nbhd}
%     Let the setup be as in Proposition~\ref{prop:inclusive_nbhd}. Then whp $\abs{\Gamma_{\CG'}(S) \setminus S} \ge \Omega(k \cdot \log^{\gamma (3-\tau)}n) - \abs{S}$.
% \end{proposition}
% \begin{proof}
%     The proposition is an immediate consequence of Proposition \ref{prop:inclusive_nbhd}. The only vertices, which can be included in $\abs{\Gamma_{\CG'}(S)}$ but not in the external neighbourhood are those in $S$ itself. Consequently $\abs{\Gamma_{\CG'}(S) \setminus S} \ge \abs{\Gamma_{\CG'}(S)} - \abs{S} \ge \Omega(k \cdot \log^{\gamma (3-\tau)}n) - \abs{S}$.
% $\hfill\qed$
% \end{proof}

We now have the tools to prove the main result. Before doing so, we begin with the next, easier result, which shows that vertex sets of constant size are $\log^{\gamma(3-\tau)} n$-expanding sets, with a by now very straightforward argument.

\begin{proposition}\label{prop:constant-size-expansion}
    Let the setup be as in Theorem \ref{thm:main-thm}. With high probability all vertex sets $S \subset \CV'$ satisfying $\abs{S} = O(1)$ have $\Omega(\log^{\gamma(3-\tau)}n)$ neighbours of weight in $[c_1'\log^\gamma n, c_2'\log^\gamma n]$.
\end{proposition}

%\MK{Adapting this proof to the new statement now.}

%\JL{A general remark on terminology: an expander is usually a graph with certain properties, not a set. I think it is fine to talk about expanding sets (it don't think it is universal terminology, but it's clear), but I would avoid the word expander for sets.}

%\US{I would put the proof (and the proposition statement as well) directly in section \ref{sec:mcd_proof}, either at the beginning or at the end (just before the proof of Theorem \ref{thm:main-thm})}
\begin{proof}
Using the same reasoning as in the proof of Proposition \ref{prop:inclusive_nbhd}, we know that whp each vertex $v\in\CV'$ has $\Theta(\log^{\gamma(3-\tau)}n)$ neighbours in $I_{i,v}$ with weights in $[c_1'\log^\gamma n, c_2'\log^\gamma n]$. This implies the statement.
\qed
\end{proof}

%Note that both Propositions~\ref{prop:inclusive_nbhd} and~\ref{prop:constant-size-expansion} also hold for $\CG''$ instead of $\CG'$ because the expected number of vertices in $\CV'' \cap I_{i,v}$ is also $\Theta(\log^{\gamma(3-\tau)}n)$, just as in $\CG'$. Hence, the expected degree $\deg_{\CV'' \cap I_{i,v}}(v)$ of $v$ restricted to vertices in $\CV'' \cap I_{i,v}$ is also $\Theta(\log^{\gamma(3-\tau)}n)$, and the rest of the argument is unchanged.

Finally, we come to the proof of our main theorem, which we restate here for convenience.

\MainTheorem*

\begin{proof}
    We distinguish between two cases based on the size of $S$. For sets of size at most $C$, where $C$ is a constant to be chosen later, the theorem holds due to Proposition \ref{prop:constant-size-expansion}. Thus, for the remainder of the proof, we may restrict to sets of size at least $C$. We also condition on the events described in Propositions \ref{prop:strip-covering} and \ref{prop:inclusive_nbhd} to hold.

    Let $S\subset \CV'$ be such that $C \le \abs{S} \le \eps\abs{\CV'}$, where $\eps\in(0,1)$ is a constant that we will fix later. Proposition \ref{prop:strip-covering} holds for all $|S|=\omega(1)$, hence it also holds for $|S|\ge C$ if $C$ is large enough. By Proposition \ref{prop:strip-covering}, there is a coordinate $i$ and a natural number $k$ satisfying
    \begin{align*}
        k = \abs{S} \cdot \Omega\br{\min\left\{1, \frac{1}{\log^{\gamma(3-\tau)}n} \cdot \br{\frac{|\CV'|}{\abs{S}}}^{1-1/c_d}\right\}}
    \end{align*}
    such that $S$ has non-empty intersection with at least $k$ disjoint $i$-strips, where $c_d > 1$ is a constant that depends only on $d$. By Proposition \ref{prop:inclusive_nbhd}, we conclude that $\abs{\Gamma_{\CG'}(S)} \ge \Omega(k \cdot \log^{\gamma (3-\tau)}n)$. Therefore we obtain
    %In Proposition \ref{prop:inclusive_nbhd}, we established a lower bound on the size of the neighbourhood $\abs{\Gamma_{\CG'}(S)}$ within vertices of $\CG'$ for every set $S \subset \CV'$. From this we can easily derive a bound for the size of the \emph{external} neighbourhood, as the neighbourhood can include at most $\abs{S}$ vertices which the external neighbourhood may not. This bound depends on $k$, i.e. on the number of strips that $S$ covers in at least one dimension. Now, Proposition \ref{prop:strip-covering} provides a lower bound on the maximum number of strips covered in any single dimension $i$. Combining these two results, we obtain
    \begin{align*}
        |N_{ext}(S)|
        &\ge \Omega(k \cdot \log^{\gamma (3-\tau)}n) - |S|\\
        &= \Omega\br{\abs{S} \cdot \min\left\{\log^{\gamma (3-\tau)}n, \br{\frac{\abs{\CV'}}{|S|}}^{1-1/c_d}\right\}}- \aS \\
        & = \br{ \Omega \br{\min\left\{\log^{\gamma (3-\tau)}n, \br{\frac{\abs{\CV'}}{|S|}}^{1-1/c_d}\right\}}-1} \cdot \aS\\
        &\eqqcolon (\Phi(\aS, n, \abs{\CV'})-1) \cdot \aS.
    \end{align*}
    We can choose a constant $\eps>0$ so small (with respect to the constants hidden by the $\Omega$-notation above) that the assumption $\aS \le \eps \abs{\CV'}$ guarantees that $\Phi(\aS,n, \abs{\CV'}) \ge 2$, and in particular 
    \[
    \Phi(\aS,n, \abs{\CV'})-1
    = \Omega(\Phi(\aS,n, \abs{\CV'}))
    = \Omega\br{\min\left\{\log^{\gamma(3-\tau)}n, \br{\frac{\abs{\CV'}}{\aS}}^{1-1/c_d}\right\}},
    \]
    which concludes the proof for $\CG'$. If we replace $\CG'$ with the graph induced by vertices of weights in $[c_1'\log^\gamma n, c_2'\log^\gamma n]$, the proof remains identical since Propositions~\ref{prop:inclusive_nbhd} and~\ref{prop:constant-size-expansion} still hold.

    Finally, consider the graph $\CH' = (\CV_{\CH'}, \CE_{\CH'})$ consisting of all vertices of $\CG$ of degree (instead of weight) at least $c'\log^\gamma n$. Then by Lemma~\ref{lem:degree-concentration}, whp every such vertex has weight at least $C_1\log^\gamma n$ for some constant $C_1>0$. Conversely by the same lemma, whp every vertex of weight at least $C_2\log^\gamma n$ is in $\CV_{\CH'}$, for a suitable constant $C_2$. Hence, if $\CV_{1}$ is the set of vertices of weight at least $C_1\log^\gamma n$, and $\CV_{2}$ is the set of vertices of weight at least $C_2\log^\gamma n$, then $\CV_1\subseteq \CV_{\CH'} \subseteq \CV_2$. Therefore, it suffices to show that every subset $S\subseteq \CV_2$ has at least $f\cdot |S|$ neighbours in $\CV_1$, where $f = f(|S|)$ is the expansion factor in~\eqref{eq:expansion-factor}. This is implied by Propositions~\ref{prop:inclusive_nbhd} and~\ref{prop:constant-size-expansion}. The proof for vertices of degrees in the interval $[c_1'\log^\gamma n, c_2'\log^\gamma n]$ is analogously, where here we use that degrees of vertices in $\CV'$ are even concentrated up to a factor $(1\pm o(1))$ by Lemma~\ref{lem:degree-concentration}.
    
$\hfill\qed$
\end{proof}